\documentclass[a4paper]{article}
\pdfoutput=1
 \usepackage{hyperref}
\usepackage[T1]{fontenc}
\usepackage{lmodern}
\usepackage{a4wide}
\usepackage{etoolbox}
\usepackage[autostyle]{csquotes}
\usepackage[protrusion=true,expansion=true]{microtype}
\usepackage[utf8]{inputenc}
\usepackage{amssymb,amsfonts, amsmath, amsthm}
\usepackage{xspace}	
\usepackage{tabularx}
\usepackage{booktabs}
\usepackage{paralist}
\usepackage{authblk}   
\usepackage{url,hyperref}
\usepackage{boxedminipage,standalone}
\usepackage{enumerate}

\usepackage{xcolor}
\definecolor{darkblue}{rgb}{0,0,0.45}
\definecolor{darkred}{rgb}{0.6,0,0}
\definecolor{darkgreen}{rgb}{0.13,0.5,0}
\usepackage{hyperref}
\hypersetup{colorlinks, linkcolor=darkblue, citecolor=darkgreen,urlcolor=darkblue} 

\usepackage{tikz}
\usetikzlibrary{shapes,fit}
\usetikzlibrary{positioning,calc}

\theoremstyle{plain}
\newtheorem{theorem}{Theorem}

\newtheorem*{remark}{Remark}
\newtheorem{corollary}{Corollary}

\newtheorem{claim}{Claim}

\newcommand{\problemdef}[3]{
	\begin{center}
		\begin{boxedminipage}{.99\textwidth}
			\textsc{{#1}}\\[2pt]
			\begin{tabular}{ r p{0.8\textwidth}}
				\textit{~~~~Instance:} & {#2}.\\
				\textit{Question:} & {#3}?
			\end{tabular}
		\end{boxedminipage}
	\end{center}
}

\newcommand{\A}{\ensuremath{A}}

\newcounter{ctrclaim}[theorem]

\newcommand{\clm}[1]{
\begin{claim}
{#1}
\end{claim}\noindent
}
\newcommand{\rmk}[1]{
\begin{remark}
{#1}
\end{remark}\noindent
}
\DeclareMathOperator{\dist}{dist}

\newcommand \dia{\hfill{$\diamond$}}

\newcommand{\yes}{{yes}\xspace}
\newcommand{\no}{{no}\xspace}

\def\phi{\varphi}

\newcommand{\NP}{{\ensuremath{\mathsf{NP}}}\xspace}

\let\phi=\varphi
\DeclareMathOperator{\df}{:=}

\title{Colouring $(P_r+P_s)$-Free Graphs\thanks{
		T.~Masa\v{r}\'{i}k, J.~Novotn\'{a}  and V.~Sl\'{i}vov\'{a} were supported by the project GAUK 1277018 and the grant SVV–2017–260452. T.~Klimo\v{s}ov\'a was supported by the Center of Excellence – ITI, project P202/12/G061 of GA \v{C}R, by the Center for Foundations of Modern Computer Science (Charles Univ. project UNCE/SCI/004), and by the project GAUK 1277018. T.~Masa\v{r}\'ik was also partly supported by the Center of Excellence – ITI, project P202/12/G061 of GA \v{C}R. V.~Slívová was partly supported by the project 17-09142S of GA \v{C}R and Charles University project PRIMUS/17/SCI/9.
D. Paulusma was supported by the Leverhulme Trust (RPG-2016-258).
An extended abstract of this paper appeared in the proceedings of ISAAC 2018~\cite{KMMNPS18}.
This version has been published in Algorithmica journal~\cite{algo} (\url{doi.org/10.1007/s00453-020-00675-w}).
}
}

\author[1]{Tereza Klimo\v{s}ov\'a}
\author[3]{Josef Mal\'ik}
\author[1,5]{Tom\'a\v{s} Masa\v{r}\'ik}
\author[1]{\\Jana Novotn\'a}
\author[4]{Dani\"el Paulusma}
\author[2]{Veronika Sl\'ivov\'a}

\affil[1]{Department of Applied Mathematics, Faculty of Mathematics and Physics,\newline Charles University, Prague, Czech Republic}
\affil[ ]{\texttt{\{tereza, masarik, janca\}@kam.mff.cuni.cz}}
\affil[2]{Computer Science Institute, Faculty of Mathematics and Physics,\newline Charles University, Prague, Czech Republic}
\affil[ ]{\texttt{slivova@iuuk.mff.cuni.cz}}
\affil[3]{Czech Technical University in Prague, Czech Republic}
\affil[ ]{\texttt{malikjo1@fit.cvut.cz}}
\affil[4]{Department of Computer Science, Durham University, Durham, UK}
\affil[ ]{\texttt{daniel.paulusma@durham.ac.uk}}
\affil[5]{Faculty of Mathematics, Informatics and Mechanics, University of Warsaw, Warsaw, Poland}

\date{}

\providecommand{\keywords}[1]{\textbf{Keywords: } #1.}

\begin{document}
\maketitle

\begin{abstract}
The {\sc $k$-Colouring} problem is to decide if the vertices of a graph can be coloured with at most $k$ colours for a fixed integer~$k$ such that no two adjacent vertices are coloured alike. If each vertex~$u$ must be assigned a colour from a prescribed list $L(u)\subseteq \{1,\ldots,k\}$, then we obtain the {\sc List $k$-Colouring} problem.
A graph~$G$ is $H$-free if $G$ does not contain $H$ as an induced subgraph. We continue an extensive study into the complexity of these two problems for $H$-free graphs.
The graph $P_r+P_s$ is the disjoint union of the $r$-vertex path~$P_r$ and the $s$-vertex path~$P_s$.
We prove that {\sc List $3$-Colouring} is polynomial-time solvable for $(P_2+P_5)$-free graphs and for $(P_3+P_4)$-free graphs.
Combining our results with known results yields complete complexity classifications of {\sc $3$-Colouring} and {\sc List $3$-Colouring} on $H$-free graphs for all graphs~$H$ up to seven vertices. 
\end{abstract}


\section{Introduction}\label{s-intro}

Graph colouring is a popular concept in Computer Science and Mathematics due to a wide range of practical and theoretical applications, as evidenced by numerous surveys and books on graph colouring and many of its variants (see, for example,~\cite{Al93,C14,GJPS,JT95,KTV99,Pa15,RS04b,Tu97}).  
Formally, a {\em colouring} of a graph $G=(V,E)$ is a mapping $c: V\rightarrow\{1,2,\ldots \}$ that assigns each vertex~$u\in V$ a {\it colour} $c(u)$ in such a way that $c(u)\neq c(v)$ whenever $uv\in E$. If $1\leq c(u)\leq k$, then $c$ is also called a {\it $k$-colouring} of $G$ and $G$ is said to be $k$-{\it colourable}. The {\sc Colouring} problem is to decide if a given graph $G$ has a $k$-colouring for some given integer~$k$.

It is well known that {\sc Colouring} is \NP-complete even if $k=3$~\cite{Lo73}. To pinpoint the reason behind the
computational hardness of {\sc Colouring} one may impose restrictions on the input. This led to an extensive study of {\sc Colouring}  for special graph classes,
particularly hereditary graph classes. A graph class is {\it hereditary} if it is closed under vertex deletion. As this is a natural property, 
hereditary graph classes capture a very large collection of well-studied graph~classes. 
A classical result in this area is due to  Gr\"otschel, Lov\'asz, and Schrijver~\cite{GLS84}, who proved that {\sc Colouring} is polynomial-time solvable for perfect graphs.

It is readily seen that
a graph class ${\cal G}$ is hereditary if and only if ${\cal G}$ can be characterized by a unique set ${\cal H}_{\cal G}$ of minimal forbidden induced subgraphs.
If ${\cal H_{\cal G}}=\{H\}$, then a graph $G\in {\cal G}$ is called {\it $H$-free}. Hence, for a graph~$H$, the class of $H$-free graphs consists of all graphs with no induced subgraph isomorphic to~$H$.

Kr\'al', Kratochv\'{\i}l, Tuza, and Woeginger~\cite{KKTW01} 
started a systematic study into the complexity of {\sc Colouring} on ${\cal H}$-free graphs for sets ${\cal H}$ of size at most~2. They showed polynomial-time solvability if $H$ is an induced subgraph of $P_4$ or $P_1+P_3$ and \NP-completeness for all other graphs~$H$. The classification for the case where ${\cal H}$ has size~2 is far from finished; see the summary in~\cite{GJPS} or an updated partial overview in~\cite{DP18} for further details.
Instead of considering sets ${\cal H}$ of size~2, we consider $H$-free graphs and follow another well-studied direction, in which the number of colours $k$ is {\it fixed}, that is, $k$ no longer belongs to the input. 
This leads to the following decision problem:

\problemdef{{\sc $k$-Colouring}}{a graph $G$}{does there exist a $k$-colouring of $G$} 

 A {\it $k$-list assignment} of $G$ is a function $L$ with domain~$V$ such that the {\it list of admissible colours} $L(u)$ of each $u\in V$ is a subset of $\{1, 2, \dots,k\}$.
A colouring $c$  {\it respects} ${L}$ if  $c(u)\in L(u)$ for every $u\in V.$ 
If $k$ is fixed, then we obtain the following generalization of $k$-{\sc Colouring}:

\problemdef{{\sc List $k$-Colouring}}{a graph $G$ and a $k$-list assignment~$L$}{does there exist a colouring of $G$ that respects $L$} 

For every $k\geq 3$, 
{\sc $k$-Colouring} on $H$-free graphs is \NP-complete if $H$ contains a cycle~\cite{EHK98} or an induced claw~\cite{Ho81,LG83}. Hence, it remains to consider the case where $H$ is a {\it linear forest} (a disjoint union of paths).
The situation is far from settled yet, although many partial results are
known~\cite{BCMSZ,BFGP13,BGPS12b,CMSZ17,CSZ,CS,CGKP15,GPS14b,HKLSS10,Hu16,LRS07,RS04a,RST02,WS01}.
Particularly, the case where~$H$ is the $t$-vertex path $P_t$ has been 
well
studied.
The cases  $k=4$, $t=7$ and  $k=5$, $t=6$ are \NP-complete~\cite{Hu16}.
For $k\geq 1$, $t=5$~\cite{HKLSS10} and $k=3$, $t=7$~\cite{BCMSZ}, even {\sc List $k$-Colouring} on $P_t$-free graphs is polynomial-time solvable (see also~\cite{GJPS}).

For a fixed integer~$k$, the $k$-{\sc Precolouring Extension} problem is to decide if a given $k$-colouring~$c'$ defined on an induced subgraph $G'$ of a graph~$G$ can be extended to a $k$-colouring~$c$ of $G$. 
Note that \textsc{$k$-Colouring} is a special case of \textsc{$k$-Precolouring Extension}, whereas the latter problem can be formulated as a special case of \textsc{List $k$-Colouring} by assigning list $\{c'(u)\}$ to every vertex~$u$ of $G'$ and list $\{1,\ldots,k\}$ to every other vertex of $G$. 
Recently, it was shown in~\cite{CSZ} that 4-{\sc Precolouring Extension}, and therefore 4-{\sc Colouring},
 is polynomial-time solvable for $P_6$-free graphs. In contrast, the more general problem {\sc List $4$-Colouring} is \NP-complete for $P_6$-free graphs~\cite{GPS14b}.  See Table~\ref{t-table1} for a summary of all these results.

\begin{table}
\vspace*{-0.45cm}
\begin{center}
\resizebox{360pt}{35pt}{
\begin{tabular}{c|c|c|c|c||c|c|c|c||c|c|c|c}
& \multicolumn{4}{c||}{{\sc $k$-Colouring}} & \multicolumn{4}{c||}{ {\sc $k$-Precolouring Extension}}&\multicolumn{4}{c}{ {\sc List $k$-Colouring}}\\
\hline
$\;\;t\;\;$                                 & {\small $k=3$}             & {\small $k=4$}                  & {\small $k=5$}               & {\small $k\ge 6$}                          & {\small $k=3$}             & {\small $k=4$}                  & {\small $k=5$}               & {\small $k\ge 6$}                             & {\small $k=3$}             & {\small $k=4$}                  & {\small $k=5$}               & {\small $k\ge 6$} \\
\hline
{\small $t\leq 5$}    & {\small P} & {\small P}      & {\small P} & {\small P}  & {\small P} & {\small P}      & {\small P} & {\small P} 
 & {\small P} & {\small P}      & {\small P} & {\small P} \\
{\small $t=6$}         & {\small P} & {\small P}                     & {\small NP-c}               & {\small NP-c}                 & {\small P} & {\small P}                    & {\small NP-c}            & {\small NP-c} 
&{\small P} & {\small NP-c}                  & {\small NP-c}            & {\small NP-c} \\
{\small $t=7$}         & 
{\small P}             
& {\small NP-c}                     & {\small NP-c}               & {\small NP-c}  & {\small P}              &  {\small NP-c}    & {\small NP-c}        & {\small NP-c}
& {\small P}              & {\small NP-c}    & {\small NP-c}        & {\small NP-c}\\
{\small $t\geq 8$}   & ?              &  {\small NP-c}    &{\small NP-c} & {\small NP-c}  &? &  {\small NP-c}    &{\small NP-c} & {\small NP-c} 
  &?              & {\small NP-c} & {\small NP-c} & {\small NP-c} 
\end{tabular}
}
\end{center}
\caption{Summary for $P_t$-free graphs.}\label{t-table1}
\end{table}

From Table~\ref{t-table1} we see that only the cases $k=3$, $t\geq 8$ are still open, although some partial results are known for $k$-{\sc Colouring} for the case $k=3$, $t=8$~\cite{CS}.
The situation when $H$ is a disconnected linear forest $\bigcup P_i$ is less clear. It is known that for every $s\geq 1$, {\sc List 3-Colouring} is polynomial-time solvable for $sP_3$-free graphs~\cite{BGPS12b,GJPS}. 
For every graph~$H$, {\sc List $3$-Colouring} is polynomial-time solvable for $(H+P_1)$-free graphs if it is polynomially solvable for $H$-free graphs~\cite{BGPS12b,GJPS}. 
If $H=rP_1+P_5$ $(r\geq 0)$, then for every integer~$k$, {\sc List $k$-Colouring} is polynomial-time solvable on $(rP_1+P_5)$-free graphs~\cite{CGKP15}. This result cannot be extended to larger linear forests~$H$, as {\sc List $4$-Colouring} is \NP-complete for $P_6$-free graphs~\cite{GPS14b}
and {\sc List $5$-Colouring} is \NP-complete for $(P_2+P_4)$-free graphs~\cite{CGKP15}. 

A way of making progress is to complete a classification by bounding the size of $H$. It follows from the above results and the ones in 
Table~\ref{t-table1} that for a graph~$H$ with $|V(H)|\leq 6$,  $3$-{\sc Colouring} and {\sc List $3$-Colouring} (and consequently,   
$3$-{\sc Precolouring Extension}) are polynomial-time solvable on $H$-free graphs if $H$ is a linear forest, and \NP-complete 
otherwise (see also~\cite{GJPS}). There are two open cases~\cite{GJPS} that must be solved in order to obtain the same statement 
for graphs~$H$ with $|V(H)|\leq 7$. These cases are
\begin{itemize}
\item [$\bullet$] $H=P_2+P_5$
\item [$\bullet$] $H=P_3+P_4$.
\end{itemize}

\subsection*{Our Results}
 
In Section~\ref{s-poly} we address the two missing cases listed above by proving the following theorem.

\begin{theorem}\label{t-main}
{\sc List $3$-Colouring} is polynomial-time solvable for $(P_2+P_5)$-free graphs and for $(P_3+P_4)$-free graphs.
\end{theorem}

We prove Theorem~\ref{t-main} as follows. If the graph $G$ of an instance $(G,L)$ of {\sc List $3$-Colouring} is $P_7$-free, then we can use the aforementioned result of Bonomo~et~al.~\cite{BCMSZ}. Hence we may assume that $G$ contains an induced $P_7$. We consider every possibility of colouring the vertices of this $P_7$ and try to reduce each resulting instance to a polynomial number of smaller instances of $2$-{\sc Satisfiability}. As the latter problem can be solved in polynomial time, the total running time of the algorithm will be polynomial. The crucial proof ingredient is that we partition the set of vertices of $G$ that do not belong to the $P_7$ into subsets of vertices that are of the same distance to the $P_7$. This leads to several ``layers'' of $G$. We analyse  how the vertices of each layer are connected to each other and to vertices of adjacent layers so as to use this information in the design of our algorithm. 

Combining Theorem~\ref{t-main} with the known results yields the following complexity classifications for graphs~$H$ up to seven vertices; 
see Section~\ref{a-a} for its proof.

\begin{corollary}\label{c-summary}
Let $H$ be a graph with $|V(H)|\leq 7$. If $H$ is a linear forest, then {\sc List $3$-Colouring} is polynomial-time solvable for $H$-free graphs; otherwise already $3$-{\sc Colouring} is \NP-complete for $H$-free graphs.
\end{corollary}

\subsection*{Preliminaries}\label{s-pre}

Let $G=(V,E)$ be a graph. 
For a vertex $v\in V$, we denote its \emph{neighbourhood} by $N(v)=\{u\; |\; uv\in E\}$, its 
{\it closed neighbourhood} by
$N[v]=N(v)\cup \{v\}$
and its degree  by $\deg(v)=|N(v)|$.
For a set $S\subseteq V$, we write $N(S)=\bigcup_{v\in S}N(v)\setminus S$ and $N[S]=N(S)\cup S$, and we let $G[S]=(S,\{uv\; |\; u,v\in S\})$ be the subgraph of $G$ induced by $S$.
The {\it contraction} of an edge $e=uv$ removes $u$ and $v$ from $G$ and introduces a new vertex which is made adjacent to every vertex in $N(u)\cup N(v)$.
 The {\it identification} of a set $S\subseteq V$ by a vertex~$w$ removes all vertices of $S$ from $G$, introduces $w$ as a new vertex and makes $w$ adjacent to every vertex in $N(S)$.
The {\it length} of a  path is its number of edges.
The {\it distance} $\dist_G(u,v)$ between two vertices $u$ and~$v$ is the length of a shortest path between them in $G$.
The {\it distance} $\dist_G(u,S)$ between a vertex $u\in V$ and a set $S\subseteq V\setminus \{v\}$ is defined as $\min\{\dist(u,v)\; |\; v\in S\}$.

For two graphs $G$ and $H$, we use $G+H$ to denote the disjoint union of $G$ and $H$, and we write $rG$ to denote the disjoint union of $r$ copies of $G$.
Let $(G,L)$ be an instance of {\sc List 3-Colouring}. For $S\subseteq V(G)$, we write $L(S)=\bigcup_{u\in S}L(u)$.
We let $P_n$ and $K_n$ denote the path and complete graph on $n$ vertices, respectively. 
The {\it diamond} is the graph obtained from $K_4$ after removing an edge.

We say that an instance $(G',L')$ is {\it smaller} than some other instance $(G,L)$ of {\sc List 3-Colouring} if either $G'$ is an induced subgraph of $G$ with $|V(G')|<|V(G)|$; or $G'=G$ and $L'(u)\subseteq L(u)$ for each $u\in V(G)$, such that there exists at least one vertex~$u^*$ with $L'(u^*)\subset L(u^*)$.

\section{The Proof of Theorem~1}\label{s-poly}

In this section we show that {\sc List $3$-Colouring} problem  is polynomial-time solvable for $(P_2+P_5)$-free graphs and for $(P_3+P_4)$-free graphs. As arguments for these two graph classes are overlapping, we prove both cases simultaneously.
Our proof uses the following two results.

\begin{theorem}[\cite{BCMSZ}]\label{t-p7}
{\sc List $3$-Colouring} is polynomial-time solvable for $P_7$-free graphs.
\end{theorem}

If we cannot apply Theorem~\ref{t-p7}, our strategy is to reduce, in polynomial time, an instance $(G,L)$ of {\sc List 3-Colouring} to a polynomial number of smaller instances of {\sc 2-List Colouring}. We use the following well-known result due to Edwards.

\begin{theorem}[\cite{Ed86}]\label{t-2sat}
The {\sc $2$-List Colouring} problem is linear-time solvable.
\end{theorem}

We are now ready to prove our main result
, namely that {\sc List $3$-Colouring} is polynomial-time solvable for $(P_2+P_5)$-free graphs and for $(P_3+P_4)$-free graphs. As arguments for these two graph classes are overlapping, we prove both cases simultaneously. We start with an outline followed by a formal proof.

\medskip
\noindent
{\it Outline of the proof of Theorem~\ref{t-main}.} Our goal is to reduce, in polynomial time, a given instance $(G,L)$ of {\sc List 3-Colouring}, where $G$ is $(P_2+P_5)$-free or $(P_3+P_4)$-free, to a polynomial number of smaller instances of {\sc 2-List-Colouring} in such a way that $(G,L)$ is a yes-instance if and only if at least one of the new instances is a yes-instance.  As for each of the smaller instances, we can apply Theorem~\ref{t-2sat}, the total running time of our algorithm will be polynomial. 

If $G$ is $P_7$-free, then we do not have to do the above and may apply Theorem~\ref{t-p7} instead. Hence, we assume that $G$ contains an induced $P_7$. We put the vertices of the $P_7$ in a set~$N_0$ and define sets $N_i$ $(i\geq 1$) of vertices of the same distance~$i$ from $N_0$; we say that the sets~$N_i$ are the layers of $G$. We then analyse the structure of these layers using the fact that $G$ is $(P_2+P_5)$-free or $(P_3+P_4)$-free. The first phase of our algorithm is about preprocessing $(G,L)$ after colouring the seven vertices of~$N_0$ and applying a number of propagation rules. We consider every possible colouring of the vertices of $N_0$. In each branch, we may have to deal with vertices $u$ that still have a list $L(u)$ of size~3. We call such vertices active and prove that they all belong to $N_2$. We then enter the second phase of our algorithm. In this phase we show, via some further branching, that $N_1$-neighbours of active vertices either all have a list from $\{\{h,i\},\{h,j\}\}$, where $\{h,i,j\}=\{1,2,3\}$, or they all have
the same list $\{h,i\}$. In the third phase, we reduce, again via some branching, to the situation where only the latter option applies: $N_1$-neighbours of active vertices all have the same list.  Then in the fourth and final phase of our algorithm, we know so much structure of the instance that we can reduce to a polynomial number of smaller instances of
{\sc 2-List-Colouring} via a new propagation rule identifying common neighbourhoods of two vertices by a single vertex.

\medskip
\noindent
{\bf Theorem~\ref{t-main} (restated).}
{\it {\sc List $3$-Colouring} is polynomial-time solvable for $(P_2+P_5)$-free graphs and for $(P_3+P_4)$-free graphs.}

\begin{proof}
\setcounter{ctrclaim}{0}
Let $(G,L)$ be an instance of {\sc List 3-Colouring}, where $G=(V,E)$ is an $H$-free graph for $H\in \{P_2+P_5,P_3+P_4\}$. Note that $G$ is $(P_3+P_5)$-free.
Since the problem can be solved component-wise, we may assume that~$G$ is connected.
If $G$ contains a $K_4$, then $G$ is not 3-colourable, and thus $(G,L)$ is a no-instance. As we can decide if $G$ contains a $K_4$ in $O(n^4)$ time by brute force, we  assume that from now on $G$ is $K_4$-free.
By brute force, we either deduce in $O(n^7)$ time that $G$ is $P_7$-free or we find an induced $P_7$ on vertices
$v_1,\ldots,v_7$ in that order. In the first case, we use Theorem~\ref{t-p7}. 
It remains to deal with the second case. 

\medskip
\noindent
{\bf Definition (Layers).} 
Let $N_0=\{v_1,\ldots,v_7\}.$ For $i\geq 1$, we define $N_i=\{u\: |\; \dist(u,N_0)=i\}$.
We call the sets $N_i$ $(i\geq 0)$ the {\it layers} of $G$.

\medskip
\noindent
In the remainder, we consider $N_0$ to be a fixed set of vertices. That is, we will update $(G,L)$ by applying a number of propagation rules and doing some (polynomial) branching, but we will never delete the vertices of $N_0$. This will enable us to exploit the $H$-freeness of $G$.

We show the following two claims about layers.

\clm{\label{c-v} $V=N_0\cup N_1\cup N_2\cup N_3$.}{\it Proof of Claim~\ref{c-v}.}
Suppose $N_i\neq\emptyset$ for some $i\geq 4$. As $G$ is connected, we may assume that $i=4$. Let $u_4\in N_4$. By definition, there exists two vertices $u_3\in N_3$ and $u_2\in N_2$ such that $u_2$ is adjacent to $u_3$ and $u_3$ is adjacent to $u_4$. Then $G$ has an induced $P_3+P_5$ on vertices $u_2,u_3,u_4, v_1,v_2,v_3,v_4, v_5$, a contradiction.\dia

\clm{\label{c-complete} $G[N_2\cup N_3]$ is the disjoint union of complete graphs of size at most~$3$, each containing at least one vertex of~$N_2$ (and thus at most two vertices of~$N_3$).}{\it Proof of Claim~\ref{c-complete}.}
First assume that $G[N_2\cup N_3]$ has a connected component~$D$ that is not a clique. Then $D$ contains an induced $P_3$, which together with the subgraph  $G[\{v_1,\ldots,v_5\}$] forms an induced $P_3+P_5$, a contradiction. Then the claim follows after recalling that $G$ is $K_4$-free and connected.\dia

\medskip
\noindent
We will now introduce a number of propagation rules, which run in polynomial time. 
We are going to apply these rules on $G$ {\it exhaustively}, that is, until none of the rules can be applied anymore.
Note that during this process some vertices of $G$ may be deleted (due to Rules~\ref{r-one} and~\ref{r-always}), but as mentioned we will ensure that we keep the vertices of~$N_0$, while we may update the other sets $N_i$ $(i\geq 1)$.
 We say that a propagation rule is {\it safe} if
the new instance is a \yes-instance of {\sc List 3-Colouring} if and only if the original instance is so.

\begin{enumerate}[\bf Rule 1.]
	\item \label{r-empty} {\bf(no empty lists)} If $L(u)=\emptyset$ for some $u\in V$, then return \texttt{no}.
		\item \label{r-2sat} {\bf(some lists of size~3)} If $|L(u)|\leq 2$ for every $u\in V$, then apply Theorem~\ref{t-2sat}.
		\item \label{r-disconnected} {\bf(connected graph)} If $G$ is disconnected, then solve {\sc List 3-Colouring}  on each 
instance $(D,L_D)$, where $D$ is a connected component of $G$ that does not contain $N_0$ and $L_D$ is the restriction of $L$ to $D$. If $D$ has no colouring respecting $L_D$, then return \texttt{no}; otherwise remove the vertices of $D$ from $G$.
\item \label{r-one} {\bf(no coloured vertices)} If $u\notin N_0$, $|L(u)|=1$ and $L(u)\cap L(v)= \emptyset$ for all $v\in N(u)$,  then remove $u$ from $G$.
	\item \label{r-one2} {\bf(single colour propagation)} If $u$ and $v$ are adjacent, $|L(u)|=1$, and $L(u)\subseteq L(v)$, then set $L(v) := L(v) \setminus L(u)$.
\item \label{r-diamond} {\bf(diamond colour propagation)} If $u$ and $v$ are adjacent and share two common neighbours $x$ and $y$ with $L(x)\neq L(y)$, then set $L(x):=L(x)\cap L(y)$ and $L(y):=L(x)\cap L(y)$.
	\item \label{r-hood} {\bf(twin colour propagation)} If $u$ and $v$ are non-adjacent, $N(u)\subseteq N(v)$, and $L(v)\subset L(u)$, then set $L(u):=L(v)$.
	\item \label{r-triangle} {\bf(triangle colour propagation)} If $u,v,w$ form a triangle, $|L(u)\cup L(v)|=2$ and $|L(w)|\ge2$, then set $L(w)\df L(w)\setminus (L(u)\cup L(v))$, so $|L(w)|\le1$.
	\item \label{r-always2} {\bf(no free colours)} If $|L(u) \setminus L(N(u))|\geq 1$ and $|L(u)|\geq 2$ for some $u\in V$, then set $L(u):=\{c\}$ for some $c\in L(u) \setminus L(N(u))$.
	\item \label{r-always} {\bf(no small degrees)} If $|L(u)|>|\deg(u)|$ for some $u\in V\setminus N_0$, then remove $u$ from~$G$.\end{enumerate}
As mentioned, our algorithm will branch at several stages to create a number of new but smaller instances, such that the original instance is a yes-instance if and only if at least one of the new instances is a yes-instance.
Unless we explicitly state otherwise, we {\it implicitly} assume that Rules~\ref{r-empty}--\ref{r-always} are applied exhaustively immediately after we branch 
(the reason why we may do this is shown in Claim~\ref{c-safe}).
If we apply Rule~\ref{r-empty} or~\ref{r-2sat} on a new instance, then a no-answer means that we will discard the branch. So our algorithm will only return a no-answer for the original instance $(G,L)$ if we discarded all branches. On the other hand, if we can apply Rule~\ref{r-2sat} on some new instance and obtain a yes-answer, then we can extend the obtained colouring to a colouring of $G$ that respects $L$, simply by restoring all the already coloured vertices that were removed from the graph due to the rules.
We will now state Claim~\ref{c-safe}.

\clm{\label{c-safe} Rules~\ref{r-empty}--\ref{r-always} are safe and their exhaustive application takes polynomial time.
Moreover, if we have not obtained a \yes- or \no-answer, then afterwards $G$ is a connected 
$(H,K_4)$-free graph,
such that $V=N_0\cup N_1\cup N_2\cup N_3$ and $2\leq |L(u)|\leq 3$ for every $u\in V\setminus N_0$.}{\it Proof of Claim~\ref{c-safe}.}
It is readily seen that Rules~\ref{r-empty}--\ref{r-one2} are safe.  For Rule~\ref{r-diamond}, this follows from the fact that any 3-colouring assigns $x$ and $y$ the same colour. For Rule~\ref{r-hood}, this follows from the fact that $u$ can always be recoloured with the same colour as $v$.  For Rule~\ref{r-triangle}, this follows from the fact that the colours from $L(u)\cup L(v)$ must be used on $u$ and $v$.  For Rule~\ref{r-always2}, this follows from the fact that no colour from $L(u) \setminus L(N(u))$ will be assigned to a vertex in $N(u)$. For Rule~\ref{r-always}, this follows from the fact that we always have a colour available for~$u$.

It is readily seen that applying Rules~\ref{r-empty},~\ref{r-2sat} and~\ref{r-one}--\ref{r-always} take polynomial time.  Applying Rule~\ref{r-disconnected} takes polynomial time, as each connected component of $G$ that does not contain $N_0$ is a complete graph on at most three vertices due to the $(H,K_4)$-freeness of $G$ (recall that $H=P_2+P_3$ or $H=P_3+P_4$). Each application of a rule either results in a \no-answer, a \yes-answer, reduces the list size of at least one vertex, or reduces $G$ by at least one vertex. Thus the exhaustive application of the rules takes polynomial time.

Suppose exhaustive application does not yield a \no-answer or a \yes-answer. By Rule~\ref{r-disconnected}, $G$ is connected.  As no vertex of $N_0$ was removed, $G$ contains $N_0$. Hence, we can define $V=N_0\cup N_1\cup N_2\cup N_3$ by Claim~\ref{c-v}. By Rules~\ref{r-one} and~\ref{r-one2},  we find that $2\leq |L(u)|\leq 3$ for every $u\in V\setminus N_0$. It is readily seen that Rules~\ref{r-empty}--\ref{r-always} preserve  $(H,K_4)$-freeness of $G$. \dia

\medskip
\noindent
\subsection*{\bf Phase 1. Preprocessing $\mathbf{(G,L)}$}

\medskip
\noindent
In Phase 1 we will preprocess $(G,L)$ using the above propagation rules. 
To start off the preprocessing we will branch via colouring the vertices of $N_0$ in every possible way. By colouring a vertex $u$, we mean reducing the list of permissible colours to size exactly one. (When $L(u)=\{c\}$, we consider vertex coloured by colour $c$.) Thus, when we colour some vertex $u$, we always give $u$ a colour from its list $L(u)$.
Moreover, when we colour more than one vertex we will always assign distinct colours to adjacent vertices.

\medskip
\noindent
{\bf Branching I} ($O(1)$ branches)\\
We now consider all possible combinations of colours that can be assigned to the vertices in $N_0$. That is, we branch into at most 
$2\cdot 3^6$ cases, in which $v_1, \ldots, v_7$ each receives a colour from their list.
We note that each branch leads to a smaller instance and that $(G,L)$ is a yes-instance if and only if at least one of the new instances is a yes-instance.
Hence, if we applied Rule~\ref{r-empty} in some branch, then we discard the branch. If we applied Rule~\ref{r-2sat} and obtained a no-answer, then we discard the branch as well. If we obtained a yes-answer, then we are done. Otherwise, we continue by considering each remaining branch separately. For each remaining branch, we denote the resulting smaller instance by $(G,L)$ again.

We will now introduce a new rule, namely Rule~\ref{r-third}. We apply Rule~\ref{r-third} together with the other rules. That is, we now apply Rules~\ref{r-empty}--\ref{r-third} exhaustively. However, each time we apply Rule~\ref{r-third} we first ensure that  Rules~\ref{r-empty}--\ref{r-always} have been applied exhaustively.

\begin{enumerate}[\bf Rule 1]
\setcounter{enumi}{10}
\item \label{r-third}  {\bf ($\mathbf{N_3}$-reduction)} If $u$ and $v$ are in $N_3$ and are adjacent, then remove $u$ and $v$
 from~$G$.
\end{enumerate}
\clm{\label{c-safethird} Rule~\ref{r-third}, applied after exhaustive application of Rules~\ref{r-empty}--\ref{r-always}, is safe and takes polynomial time.
Moreover, afterwards $G$ is a connected 
$(H,K_4)$-free graph,
such that $V=N_0\cup N_1\cup N_2\cup N_3$ and $2\leq |L(u)|\leq 3$ for every $u\in V\setminus N_0$.}{\it Proof of Claim~\ref{c-safethird}.}
Assume that we applied Rules~\ref{r-empty}--\ref{r-always} exhaustively and that $N_3$ contains two adjacent vertices $u$ and $v$.  By Claim~\ref{c-complete}, we find that $u$ and $v$ have a common neighbour~$w\in N_2$ and no other neighbours. By Rules~\ref{r-one},~\ref{r-one2} and~\ref{r-always}, we then find that $|L(u)|=|L(v)|=2$. First suppose that $L(u)=L(v)$, say $L(u)=L(v)=\{1,2\}$. Then, by Rule~\ref{r-triangle}, we find that $L(w)=\{3\}$, contradicting Rule~\ref{r-one}.  Hence $L(u)\neq L(v)$, say $L(u)=\{1,2\}$ and $L(v)=\{1,3\}$. By Rule~\ref{r-triangle}, we find that $L(w)=\{2,3\}$ or $L(w)=\{1,2,3\}$. If $w$ gets colour~1, we can give $u$ colour~2 and $v$ colour~3. If $w$ gets colour~2, we can give $u$ colour~1 and $v$ colour~3. Finally, if $w$ gets colour~3, then we can give $u$ colour~2 and $v$ colour~1. Hence we may set $V:=V\setminus \{u,v\}$. This does not destroy the connectivity or $(H,K_4)$-freeness of $G$.\dia

\medskip
\noindent
We now show the following claim.

\clm{\label{c-independent} The set $N_3$ is independent, and moreover, each vertex $u\in N_3$ has $|L(u)|=2$ and exactly two neighbours in $N_2$ which are adjacent.}{\it Proof of Claim~\ref{c-independent}.} 
By Rule~\ref{r-third}, we find that $N_3$ is independent. 
By Claim~\ref{c-complete}, every vertex of $N_3$ has at most two neighbours in $N_2$ and these neighbours are adjacent.
Hence, the claim follows from Rules~\ref{r-one},~\ref{r-one2},~\ref{r-always} and the fact that $N_3$ is independent.
\dia
\begin{figure}[tb]
	\begin{center}
		\includestandalone[mode=image,scale=0.7]{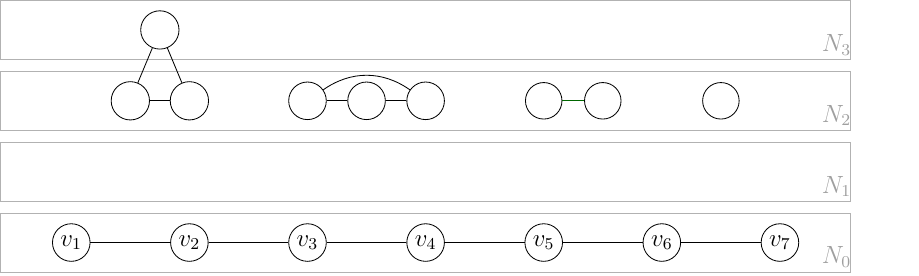}
		\caption{All possible connected components in $G[N_2\cup N_3]$.}\label{f-structure}
	\end{center}
		\small
\end{figure}

\medskip
\noindent
The following claim is an immediate consequence of Claims~\ref{c-complete} and~\ref{c-independent} and gives a complete description of the second and third layer, see also Figure~\ref{f-structure}.

\clm{\label{c-layers} Every connected component $D$ of 
$G[N_2\cup N_3]$ is a complete graph with either $|D|\leq 2$ and $D\subseteq N_2$, or $|D|=3$ and $|D\cap N_3|\leq 1$.}

\noindent
The following claim describes the location of the vertices with list of size~3 in $G$.

\clm{\label{c-n2} For every $u\in V$, if $|L(u)|=3$, then $u\in N_2$.}{\it Proof of Claim~\ref{c-n2}.}
As the vertices in $N_0$ have lists of size~1, the vertices in $N_1$ have lists of size~2.  By Claim~\ref{c-independent}, the same holds for vertices in $N_3$.\dia

\medskip
\noindent
In the remainder of the proof, we will show how to branch in order to reduce
the lists of the vertices $u\in N_2$ with $|L(u)|=3$ by at least one colour. 
We formalize this approach in the following definition.

\medskip
\noindent
{\bf Definition (Active vertices).}
A vertex $u\in N_2$ and its neighbours in $N_1$ are called \emph{active} if $|L(u)|=3$.
Let $\A$ be the set of all active vertices.  Let $\A_1=\A\cap N_1$ and $\A_2= \A\cap N_2$. 
We \emph{deactivate} a vertex $u\in A_2$ if we reduce the list $L(u)$ by at least one colour. We \emph{deactivate} a vertex
$w\in A_1$ by deactivating all its neighbours in $A_2$. 

\medskip
\noindent
Note that every vertex $w\in A_1$ has $|L(w)|=2$ by Rule~\ref{r-one2} applied on the vertices of $N_0$. Hence, if we reduce $L(w)$ by one colour, all neighbours of $w$ in $A_2$ become deactivated by Rule~\ref{r-one2}, and $w$ is removed by Rule~\ref{r-one}.

For $1\leq i<j\leq 7$, we let $A(i,j)\subseteq A_1$ be the set of active neighbours of $v_i$ that are not adjacent to $v_j$ and similarly, we let $A(j,i)\subseteq A_1$ be the set of active neighbours of $v_j$ that are not adjacent to $v_i$.

\medskip
\noindent
\subsection*{\bf Phase 2. Reduce the number of distinct sets $\mathbf{A(i,j)}$}

\medskip
\noindent
We will now branch into $O(n^{45})$ smaller instances such that $(G,L)$ is a yes-instance of {\sc List 3-Colouring} if and only if at least one of these new instances is a yes-instance. Each new instance will have the following property:

\begin{description}
\item[(P\label{c-first*})]\label{cond:priv-act*} for $1\leq i\leq j\leq 7$ with $j-i\geq 2$, either $A(i,j)=\emptyset$ or $A(j,i)=\emptyset$.
\end{description}

\noindent
{\bf Branching II} ($O\bigl(n^{\bigl(3\cdot\left(\binom{7}{2}-6\right)\bigr)}\bigr)=O(n^{45})$ branches)\\
Consider two vertices $v_i$ and $v_j$ with $1\leq i\leq j\leq 7$ and $j-i\geq 2$.
Assume without loss of generality that $v_i$ is coloured~3 and that $v_j$ is coloured either~1 or~3. 
Hence, every $w\in A(i,j)$ has 
$L(w)=\{1,2\}$, whereas every $w\in A(j,i)$ has $L(w)=\{2,q\}$ for $q\in \{1,3\}$.
We branch as follows.
We consider all possibilities where at most one vertex of $A(i,j)$ receives colour~2 (and all other vertices of $A(i,j)$ receive colour~1) and all possibilities where we choose two vertices from $A(i,j)$ to receive colour~2.
This leads to $O(n)+O(n^2)=O(n^2)$ branches.
In the branches where at most one vertex of $A(i,j)$ receives colour~2,  every vertex of $A(i,j)$ will be deactivated.
So Property~{\bf (P)} is satisfied for $i$ and $j$. 

Now consider the branches where two vertices $x_1,x_2$ of $A(i,j)$ both received colour~2.  We update $A(j,i)$ accordingly. In particular, afterwards no vertex in $A(j,i)$ is adjacent to $x_1$ or $x_2$, as $2$ is a colour in the list of each vertex of $A(j,i)$. 
We now do some further branching for those branches where $A(j,i)\neq \emptyset$. 
We consider the possibility where each vertex of $N(A(j,i))\cap A_2$ is given the colour of $v_j$ and all possibilities where we choose one vertex in $N(A(j,i))\cap A_2$ to receive a colour different from the colour of $v_j$ (we consider both options to colour such a vertex).
This leads to $O(n)$ branches. In the first branch, every vertex of $A(j,i)$ will be deactivated.
So Property~{\bf (P)} is satisfied for $i$ and $j$. 

Now consider a branch where a vertex $u\in N(A(j,i))\cap A_2$ receives a colour different from the colour of $v_j$. 
We will show that also, in this case, every vertex of $A(j,i)$ will be deactivated.
For contradiction, assume that $A(j,i)$ contains a vertex $w$ that is not deactivated after colouring~$u$. 
As $u$ was in $N(A(j,i))\cap A_2$, we find that $u$ had a neighbour $w'\in A(j,i)$. 
As $u$ is coloured with a colour different from the colour of $v_j$, the size of $L(w')$ is reduced by one (due to Rule~\ref{r-one}). Hence $w'$ got deactivated after colouring~$u$, and thus $w'\neq w$. As $w$ is still active, $w$ has a neighbour $u'\in A_2$. As $u'$ and $w$ are still active, $u'$ and $w$ are not adjacent to $w'$ or $u$.
 Hence, $u,w',v_j,w,u'$ induce a $P_5$ in $G$. As $x_1$ and $x_2$ both received colour~2, we find that $x_1$ and $x_2$ are not adjacent to each other.
Hence, $x_1,v_i,x_2$ induce a $P_3$ in $G$. Recall that all vertices of $A(j,i)$, so also $w$ and $w'$, are not adjacent to $x_1$ or $x_2$.  As $u$ and $u'$ were still active after colouring $x_1$ and $x_2$, we find that $u$ and $u'$ are not adjacent to $x_1$ or $x_2$ either.
By definition of $A(j,i)$, $w$ and $w'$ are not adjacent to $v_i$. By definition of $A(i,j)$, $x_1$ and $x_2$ are not adjacent to $v_j$.
 Moreover, $v_i$ and $v_j$ are non-adjacent, as $j-i\geq 2$. We conclude that $G$ contains an induced $P_3+P_5$, namely with vertex 
 set $\{x_1,v_i,x_2\} \cup \{u,w',v_j,w,u'\}$, a contradiction (see Figure~\ref{f-phase2} for an example of such a situation). 
Hence, every vertex of $A(j,i)$ is deactivated. So Property~{\bf (P)} is satisfied for $i$ and $j$ also for these branches.

Finally by recursive application of the above described procedure for all pairs $v_i, v_j$ such that $1\leq i\leq j\leq 7$ and $j-i\geq 2$ we get a graph satisfying Property~{\bf (P)}, which together leads to $O\bigl(n^{\bigl(3\cdot\left(\binom{7}{2}-6\right)\bigr)}\bigr)=O(n^{45})$ branches.

\begin{figure}[tb]
	\begin{center}
		\includestandalone[mode=image,scale=0.9]{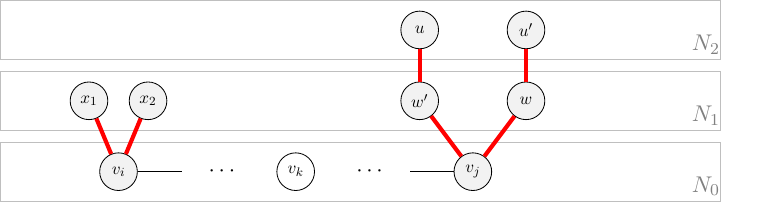}
		\caption{The situation in Branching II.}\label{f-phase2}
	\end{center}
		\small
\end{figure}

\medskip
\noindent
We now consider each resulting instance from Branching II. 
We denote such an instance by $(G,L)$ again. Note that vertices from $N_2$ may now belong to $N_3$, as their neighbours in $N_1$ may have been removed due to the branching.
The exhaustive application of Rules~\ref{r-empty}--~\ref{r-third} preserves~{\bf (P)} (where we apply
Rule~\ref{r-third} only after applying Rules~\ref{r-empty}--\ref{r-always} exhaustively). Hence $(G,L)$ satisfies~{\bf (P)}. 

We observe that if two vertices in $A_1$ have a different list, then they must be adjacent to different vertices of $N_0$. Hence, by Property~{\bf (P)}, at most two lists of $\{\{1,2\},\{1,3\},\{2,3\}\}$  can occur as lists of vertices of $A_1$. Without loss of generality this leads to two cases: either every vertex of $A_1$ has 
list $\{1,2\}$ or $\{1,3\}$ and both lists occur on $A_1$; or every vertex of $A_1$ has list $\{1,2\}$ only. 
In the next phase of our algorithm, we reduce, via some further branching, every instance of the first case to a polynomial number of smaller instances of the second case.

\medskip
\noindent

\subsection*{\bf Phase 3. Reduce to the case where vertices of $\mathbf{A_1}$ have the same list}

\medskip
\noindent
Recall that we assume that every vertex of $A_1$ has list $\{1,2\}$ or $\{1,3\}$. In this phase, we deal with the case when both types of lists occur in $A_1$. We first prove the following claim.

\clm{\label{c-oddPathVertices} 
	Let $i\in \{1,3,5,7\}$. Then every vertex from $A_1 \cap N(v_i)$ is adjacent to some vertex $v_j$ with $j \not\in \{i-1,i,i+1\}$.
}{\it Proof of Claim~\ref{c-oddPathVertices}.} 
We may assume without loss of generality that $i=1$ or $i=3$. For contradiction suppose there exists a vertex $w \in A_1 \cap N(v_i)$ that is non-adjacent to all $v_j$ with  $j \not\in \{i-1,i,i+1\}$. As two consecutive vertices in $N_0$ have different colours,  no vertex in $A_1$ has two consecutive neighbours in $N_0$ due to Rules~\ref{r-one} and~\ref{r-one2}. Hence $N(w) \cap N_0 = \{v_i\}$. By definition, $w$ has a neighbour $u\in A_2$.  If $i=1$, then $\left\{ u,w,v_1,v_2,v_3\right\} \cup \left\{ v_5,v_6,v_7 \right\}$ induces a $P_3+P_5$ in $G$. If $i=3$, then $\left\{ v_1,v_2,v_3,w,u \right\} \cup \left\{ v_5,v_6,v_7 \right\}$ induces a $P_3+P_5$ in $G$. \dia
 
\clm{\label{c-structure} It holds that $N(A_1)\cap N_0=\{v_{i-1},v_{i},v_{i+1}\}$ for some $2\leq i\leq 6$.
Moreover, we may assume without loss of generality  that $v_{i-1}$ and $v_{i+1}$ have colour~$3$ and both are adjacent to all vertices of $A_1$ with list $\{1,2\}$,
whereas $v_{i}$ has colour~$2$ and is adjacent to all vertices of $A_1$ with list $\{1,3\}$.}{\it Proof of Claim~\ref{c-structure}.} 
Recall that lists $\{1,2\}$ and $\{1,3\}$ both occur on $A_1$. 
For any two vertices $x\in A_1$ with $L(x)=\{1,2\}$ and $y\in A_1$ with $L(y)=\{1,3\}$, there exist indexes $i,j$ such that $x\in A(i,j)$
and $y\in A(j,i)$ 
(namely, $x$ is adjacent to some vertex $v_i$ with colour~$3$ and $y$ is adjacent to some vertex~$v_j$ with colour~$2$).
Note that $x$ and $y$ share no neighbour in $N_0$.
By using Property~{\bf (P)}, we find that each vertex of $N(x)\cap N_0$ must be adjacent to each vertex of $N(y)\cap N_0$. 
We conclude that either $N(A_1)\cap N_0=\{v_{i-1},v_{i}\}$ for some $2\leq i\leq 7$, or $N(A_1)\cap N_0=\{v_{i-1},v_{i},v_{i+1}\}$ for some $2\leq i\leq 6$.

The case where $N(A_1)\cap N_0=\{v_{i-1},v_{i}\}$ for some $2\leq i\leq 7$ is not possible due to Claim~\ref{c-oddPathVertices}. It follows that $N(A_1)\cap N_0=\{v_{i-1},v_{i},v_{i+1}\}$ for some $2\leq i\leq 6$. We may assume without loss of generality that $v_{i}$ has colour~$2$, meaning that $v_{i-1}$ and $v_{i+1}$ must have colour~3. It follows that every vertex of $A_1$ with list $\{1,3\}$ is adjacent to $v_i$ but not to $v_{i-1}$ or $v_{i+1}$, whereas every vertex of $A_1$ with list $\{1,2\}$ is adjacent to at least one vertex of $\{v_{i-1},v_{i+1}\}$ but not to $v_i$. As a vertex of $A_1$ with list $\{1,3\}$ has $v_i$ as its only neighbour in $N_0$, it follows from Claim~\ref{c-oddPathVertices} that $i$ is an even number. This means that $i-1$ is odd. Hence, every vertex of $A_1$ with list $\{1,2\}$ is in fact adjacent to both $v_{i-1}$ and $v_{i+1}$  due to Claim~\ref{c-oddPathVertices}.\dia

\begin{figure}[tb]
	\begin{center}
		\includestandalone[mode=image]{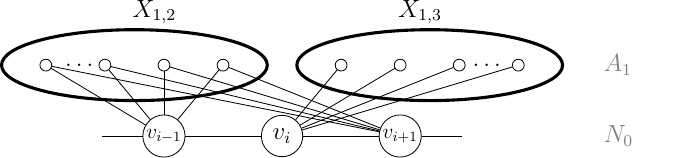}
		\caption{The situation after Claim~\ref{c-structure}.}\label{f-phase3-situation}
	\end{center}
		\small
\end{figure}

\medskip
\noindent
By Claim~\ref{c-structure}, we can partition the set $A_1$ into two (non-empty) sets
$X_{1,2}$ and $X_{1,3}$, where $X_{1,2}$ is the set of vertices in $A_1$ with list $\{1,2\}$ whose only neighbours in $N_0$ are  $v_{i-1}$ and $v_{i+1}$
(which both have colour~3) and $X_{1,3}$ is the set of vertices in $A_1$ with list $\{1,3\}$ whose only neighbour in $N_0$ is $v_i$ (which has colour~2), see Figure~\ref{f-phase3-situation}.

Our goal is to show that we can branch into at most $O(n^2)$ smaller instances, in which either $X_{1,2}=\emptyset$ or $X_{1,3}=\emptyset$, such that $(G,L)$ is a yes-instance of 
{\sc List 3-Colouring} if and only if at least one of these smaller instances is a yes-instance. Then afterwards it suffices to show how to deal with the case where all vertices in $A_1$ have the same list in polynomial time; this will be done in Phase~4 of the algorithm. We start with the following $O(n)$ branching procedure (in each of the branches we may do some further $O(n)$ branching later on).

\medskip
\noindent
{\bf Branching III} ($O(n)$ branches)\\
We branch by considering the possibility of giving each vertex in $X_{1,2}$ colour~2 and all possibilities of choosing a vertex in $X_{1,2}$ and giving it colour~1.
 This leads to $O(n)$ branches. In the first branch we obtain $X_{1,2}=\emptyset$. Hence
we can start Phase 4 for this branch. We now consider every branch in which $X_{1,2}$ and $X_{1,3}$ are both nonempty. 
For each such branch we will create $O(n)$ smaller instances of {\sc List 3-Colouring}, where $X_{1,3}=\emptyset$, such
that $(G,L)$ is a yes-instance of {\sc List 3-Colouring} if and only if at least one of the new instances is a yes-instance.

Let $w\in X_{1,2}$ be the vertex that was given colour~1 in such a branch. Although by Rule~\ref{r-one} vertex~$w$ will need to be removed from $G$, we make an exception by temporarily keeping $w$ after we coloured it. The reason is that the presence of $w$ will be helpful for analysing the structure of $(G,L)$ after Rules~\ref{r-empty}--\ref{r-third} have been applied exhaustively
(where we apply Rule~\ref{r-third} only after applying Rules~\ref{r-empty}--\ref{r-always} exhaustively). 
In order to do this, we first show the following three claims.

\clm{\label{c-p1} Vertex $w$ is not adjacent to any vertex in $A_2\cup X_{1,2}\cup X_{1,3}$.}{\it Proof of Claim~\ref{c-p1}.} 
By giving $w$ colour~1, the list of every neighbour of $w$ in $A_2$ has been reduced by one due to Rule~\ref{r-one2}. Hence, all neighbours of $w$ in $A_2$ are deactivated.  For the same reason all neighbours of $w$ in $X_{1,2}$, which have list $\{1,2\}$, are coloured~2, and all neighbours of $w$ in $X_{1,3}$, which have list $\{1,3\}$, are coloured~3. These vertices were removed from the graph by Rule~\ref{r-one}. This proves the claim. \dia

\clm{\label{c-p2} The graph $G[X_{1,3}\cup (N(X_{1,3})\cap A_2)\cup N_3]$ is the disjoint union of one or more complete graphs, each of which consists of either one vertex of $X_{1,3}$ and at most two vertices of $A_2$, or one vertex of $N_3$.}{\it Proof of Claim~\ref{c-p2}.} 
We write $G^*=G[X_{1,3}\cup (N(X_{1,3})\cap A_2)\cup N_3]$ and first show that $G^*$ is the disjoint union of one or more complete graphs. For contradiction, assume that $G^*$ is not such a graph. Then $G^*$ contains an induced $P_3$, say on vertices $u_1,u_2,u_3$ in that order. As $w\in X_{1,2}\subseteq N_1$, we find that $w$ is not adjacent to any vertex of $N_3$. By Claim~\ref{c-p1}, we find that $w$ is not adjacent to any vertex of $A_2\cup X_{1,3}$.  Recall that $v_{i-1}$ and $v_{i+1}$ are the only neighbours of $w$ in $N_0$, whereas $v_i$ is the only neighbour of the vertices of $X_{1,3}$ in $N_0$. Hence, $\{u_1,u_2,u_3\} \cup \{v_1, \ldots, v_{i-1},w,v_{i+1},\ldots,v_7\}$ induces a $P_3+P_7$. This contradicts the $(P_3+P_5)$-freeness of $G$. We conclude that $G^*$ is the disjoint union of one or more complete graphs.

As $G$ is $K_4$-free, the above means that every connected component of $G^*$ is a complete graph on at most three vertices.
No vertex of $N_3$ is adjacent to a vertex in $X_{1,3}\subseteq N_1$. Moreover, by definition, every vertex of $N(X_{1,3})\cap A_2$ is adjacent to at least one vertex of $X_{1,3}$. As every connected component of $G^*$ is a complete graph, this means that no vertex of $N_3$ is adjacent to a vertex of $N(X_{1,3})\cap A_2$ either. We conclude that the vertices of $N_3$ are isolated vertices of $G^*$.

Let $D$ be a connected component of $G^*$ that does not contain a vertex of $N_3$. From the above we find that $D$ is a complete graph on at most three vertices. By definition, every vertex in $X_{1,3}$ has a neighbour in $A_2$ and every vertex of $N(X_{1,3})\cap A_2$ has a neighbour in $X_{1,3}$. This means that $D$ either consists of one vertex in $X_{1,3}$ and at most two vertices of $A_2$, or $D$ consists of two vertices of $X_{1,3}$ and one vertex of $A_2$. We claim that the latter case is not possible. For contradiction, assume that $D$ is a triangle that consists of three vertices $s,u_1,u_2$, where $s\in A_2$ and $u_1,u_2\in X_{1,3}$. However, as $L(u_1)=L(u_2)=\{1,3\}$, we find that $|L(s)|=1$ by Rule~\ref{r-triangle}, contradicting the fact that $s$ belongs to $A_2$. This completes the proof of the claim. \dia

\clm{\label{c-p3} For every pair of adjacent vertices $s,t$ with $s\in A_2$ and $t\in N_2$, either $t$ is adjacent to $w$, or 
$N(s) \cap X_{1,3} \subseteq N(t)$.}
\textit{Proof of Claim~\ref{c-p3}.}
For contradiction, assume that $t$ is not adjacent to $w$ and that there is a vertex $r \in X_{1,3}$ that is adjacent to $s$ but not to $t$. By Claim~\ref{c-p1}, we find that $w$ is not adjacent to $r$ or $s$. Just as in the proof of Claim~\ref{c-p2}, we find that $\{r,s,t\}$ together with  $\{v_1,,\ldots,v_{i-1},w,v_{i+1},\ldots,v_7\}$ induces a $P_3+P_7$ in $G$, a contradiction. \dia

\medskip
\noindent
We now continue as follows. Recall that $X_{1,3}\neq \emptyset$. Hence there exists a vertex $s\in A_2$ that has a neighbour $r\in X_{1,3}$. As $s\in A_2$, we have that $|L(s)|=3$.
Then, by Rule~\ref{r-always}, we find that $s$ has at least two neighbours $t$ and $t'$ not equal to $r$. By Claim~\ref{c-p2},
we find that neither $t$ nor $t'$ belongs to $X_{1,3}\cup N_3$.
We are going to fix an induced 3-vertex path~$P^s$ of $G$, over which we will branch, in the following way.

If $t$ and $t'$ are not adjacent, then we let $P^s$ be the induced path in $G$ with vertices $t,s,t'$ in that order.
Suppose that $t$ and $t'$ are adjacent. As $G$ is $K_4$-free and $s$ is adjacent to $r,t,t'$, at least one of $t,t'$ is not adjacent to $r$. We may assume without loss of generality that $t$ is not adjacent to $r$.

First assume that $t\in N_2$. Recall that $s$ has a neighbour in $X_{1,3}$, namely~$r$, and that $r$ is not adjacent to $t$. We then find that $t$ must be adjacent to $w$ by Claim~\ref{c-p3}. As $s\in A_2$, we find that $s$ is not adjacent to $w$ by 
Claim~\ref{c-p1}. In this case we let $P^s$ be the induced path in $G$ with vertices $s,t,w$ in that order.

Now assume that $t\notin N_2$. Recall that $t\notin N_3$. Hence, $t$ must be in $N_1$. Then, as 
$t\notin X_{1,3}$ but $t$ is adjacent to a vertex in $A_2$, namely $s$, we find that $t\in X_{1,2}$. 
Recall that $t'\notin X_{1,3}$. If $t'\in N_1$ then the fact that $t'\notin X_{1,3}$, combined with the fact that $t'$ is adjacent to $s\in A_2$, implies that $t'\in X_{1,2}$. However, by Rule~\ref{r-triangle} applied on $s,t,t'$, vertex~$s$ would have a list of size~1 instead of size~3, a contradiction. Hence, $t'\notin N_1$. As $t'\notin N_3$, this means that $t'\in N_2$. If $t'$ is adjacent to $r$, then $t\in X_{1,2}$ with $L(t)=\{1,2\}$ and $r\in X_{1,3}$ with $L(r)=\{1,3\}$ would have the same lists by Rule~\ref{r-diamond} applied on $r,s,t,t'$, a contradiction. Hence $t'$ is not adjacent to $r$.
Then, by Claim~\ref{c-p3}, we find that $t'$ must be adjacent to $w$. Note that $s$ is not adjacent to $w$ due to Claim~\ref{c-p1}.
In this case we let $P^s$ be the induced path in $G$ with vertices $s,t',w$ in that order.

We conclude that either $P^s=tst'$ or $P^s=stw$ or $P^s=st'w$. We are now ready to apply another round of branching.

\medskip
\noindent
{\bf Branching IV} ($O(n)$ branches)\\
We branch by considering the possibility of removing colour~2 from the list of each vertex in 
$N(X_{1,3})\cap A_2$
 and all possibilities of choosing a vertex in 
 $N(X_{1,3})\cap A_2$ and giving it colour~2. 
In the branch where we removed colour~2 from the list of every
vertex in 
 $N(X_{1,3})\cap A_2$, we obtain that $X_{1,3}=\emptyset$. Hence for that branch we can enter Phase~4.
Now consider a branch where we gave some vertex 
$s\in N(X_{1,3})\cap A_2$ colour~2. 
Let $P^{s}=tst'$ or $P^{s}=stw$ or $P^{s}=st'w$.
We do some further branching by considering all possibilities of colouring the vertices of $P^{s}$ that are not equal to the already coloured vertices $s$ and $w$ (should $w$ be a vertex of $P^{s}$) 
and all possibilities of giving a colour to the vertex from $N(s) \cap X_{1,3}$ (recall that by Claim~\ref{c-p2}, $|N(s) \cap X_{1,3}|=1$).
This leads to a total of $O(n)$ branches.
We claim that in each of these branches, the size of $X_{1,3}$ has reduced to at most~1. 

\begin{figure}[tb]
	\begin{center}
		\includestandalone[mode=image,scale=0.9]{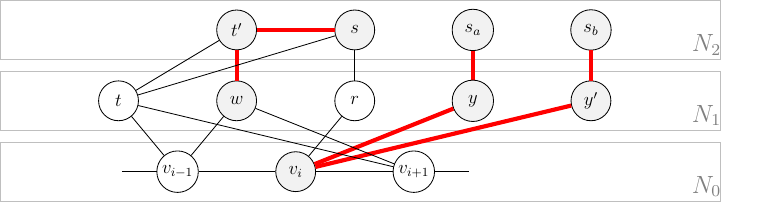}
		\caption{The situation in Branching IV if $t_1\in N_1$ and if vertices $s_a$ and $s_b$ exist.}\label{f-phase3}
	\end{center}
		\small
\end{figure}

For contradiction, assume that there exists a branch where $X_{1,3}$ contains two vertices $y$ and $y'$.
Let $s_a$ and $s_b$ be the neighbours of $y$ and $y'$ in $A_2$, respectively. By Claim~\ref{c-p2}, the graph induced by
$\{y,y',s_a,s_b\}$ is isomorphic to $2P_2$. Hence, the set $\{s_a, y, v_{i}, y', s_b\}$ induces a $P_5$ in $G$.
Recall that $P^{s}=tst'$ or $P^{s}=stw$ or $P^{s}=st'w$.
As $s_a$ and $s_b$ have a list of size~3, neither $s_a$ nor $s_b$ is adjacent to a vertex of $P^{s}$ due to Rule~\ref{r-one2}.
Neither $y$ nor $y'$ is adjacent to $N(s)\cap X_{1,3}$, as $N(s)\cap X_{1,3}$ is already coloured.
By Claims~\ref{c-p1} and~\ref{c-p2}, neither $y$ nor $y'$ is adjacent to $w$ or $s$, respectively.
As $s$ received colour~2,
vertices $t$ and $t'$ have received colour~1 or~3 should they belong to $P^{s}$. In that case neither $t$ nor $t'$ can be adjacent to $y$ or $y'$, as $L(y)=L(y')=\{1,3\}$. By definition, $v_{i}$ is not adjacent to $s$ or $w$. Moreover, $v_{i}$ can only be adjacent to a vertex from $\{t,t'\}$ if  that vertex belonged to $N_1$. However, recall that $t$ and $t'$ were not in $X_{1,3}$ while $s$ was an active vertex. Hence if $t$ or $t'$ belonged to $N_1$, they must have been in $X_{1,2}$ and thus not adjacent to $v_{i}$. This means that the vertices of $P^{s}$, together with $\{s_a,y,v_{i},y',s_b\}$, induce a $P_3+P_5$ in $G$, a contradiction (see Figure~\ref{f-phase3} for an example of such a situation). 
Thus $X_{1,3}$ must contain at most one vertex.

\medskip
\noindent
{\bf Branching V} ($O(1)$ branches)\\	
We branch by considering both possibilities of colouring the unique vertex of $X_{1,3}$. This leads to two new but smaller instances of {\sc List 3-Colouring}, in each of which the set $X_{1,3}=\emptyset$. Hence, our algorithm can enter Phase~4.

\medskip
\noindent

\subsection*{\bf Phase 4. Reduce to a set of instances of 2-List Colouring}

\medskip
\noindent
Recall that in this stage of our algorithm we have an instance $(G,L)$ in which 
every vertex of $A_1$ has the same list, say $\{1,2\}$. 
We deal with this case as follows. First suppose that $H=P_2+P_5$. Then $G[N_2\cup N_3]$ is an independent set, as otherwise two adjacent vertices of $N_2\cup N_3$ form, together with $v_1,\ldots,v_5$, an induced $P_2+P_5$. Hence, we can safely colour each vertex in $A_2$ with colour~3, and afterwards we may apply Theorem~\ref{t-2sat}.

Now suppose that $H=P_3+P_4$.
We first introduce two new rules, which turn $(G,L)$ into a smaller instance. In Claims~\ref{c-safe2} and~\ref{c-uNeighbors} we show that we may include those rules in our set of propagation rules that we apply implicitly every time we modify the instance $(G,L)$.

\begin{enumerate}[\bf Rule 1]
\setcounter{enumi}{11}
\item \label{r-identify} {\bf (neighbourhood identification)} If $u$ and $v$ are adjacent, $N(v)\subseteq N[u]$, $N(u)\cap N(v)\neq \emptyset$, and $|L(v)|=3$, then identify 
$N(u)\cap N(v)$ by $w$, set $L(w):=\bigcap\{L(x)\; |\;x\in N(u)\cap N(v)\}$ and remove $v$ from $G$.
If $G$ contains a $K_4$, then return {\tt no}.
\end{enumerate}

\noindent
We note that the case where $u$ and $v$ are adjacent, $N(v)\subseteq N[u]$, and $N(u)\cap N(v)= \emptyset$ implies that $N(v)=\{u\}$, and thus $\deg(v)=1$.
Therefore, this case was already handled by one of the Rules~\ref{r-empty},~\ref{r-one}--\ref{r-one2}, or~\ref{r-always}.
Whenever we refer to Rule~\ref{r-identify} we always assume that the previous rules were applied meaning that we will 
implicitly assume that $N(u)\cap N(v)\neq \emptyset$.

\clm{\label{c-safe2} Rule~\ref{r-identify} is safe for $K_4$-free input, takes polynomial time and does not affect any vertex of $N_0$.
Moreover, if we have not obtained a \no-answer, then afterwards $G$ is a connected $(H, K_4)$-free graph, in which
we can define sets $N_1,N_2,N_3,A_1,A_2$ as before.}{\it Proof of Claim~\ref{c-safe2}.} 
Note that by Claim~\ref{c-safe}, $G$ is $K_4$-free before the application of Rule~\ref{r-identify}. Hence $N(u)\cap N(v)$ is an independent set. Let $w$ be the new vertex obtained from identifying $N(u)\cap N(v)$. Observe that every vertex in the common neighbourhood of two adjacent vertices must receive the same colour. Hence $w$ can be given the same colour as any vertex of $N(u)\cap N(v)$, which belongs to $\bigcap\{L(x)\; |\;x\in N(u)\cap N(v)\}$. For the reverse direction, we give each vertex $x\in N(u)\cap N(v)$ the colour of $w$, which belongs to $L(x)$ by definition. As $|L(v)|=3$ and $N(v)\setminus N(u)=\{u\}$, we have a colour available for $v$.  The above means that $(G,L)$ is a no-instance if a $K_4$ is created. We conclude that Rule~\ref{r-identify} is safe and either yields a no-instance if a $K_4$ was created, or afterwards we have again that $G$ is $K_4$-free.

It is readily seen that applying Rule~\ref{r-identify} takes polynomial time and that afterwards $G$ is still connected. As $|L(v)|=3$, Claim~\ref{c-n2} tells us that $v\in N_2$, and thus $N(v)\subseteq N_1\cup N_2\cup N_3$. Thus Rule~\ref{r-identify} does not involve any vertex of $N_0$. Hence, as $G$ is connected, we can define $V=N_0\cup N_1\cup N_2\cup N_3$ by Claim~\ref{c-v}.

It remains to prove that $G$ is $H$-free after applying Rule~\ref{r-identify}.  For contradiction, assume that $G$ has an induced subgraph $P+P'$ isomorphic to $H$. 
Then we find that the vertex $w$ created by Rule~\ref{r-identify} must be in $V(P)\cup V(P')$,
as otherwise, $P+P'$ was already an induced subgraph of $G$ before Rule~\ref{r-identify} was applied.
We assume, without loss of generality, that $w$ belongs to $V(P)$.
By the same argument, we find that $w$ is incident with two edges $wx$ and $wy$ in $P$ that correspond to edges $sx$ and $ty$ with $s\neq t$ in~$G$ before Rule~\ref{r-identify} was applied 
(where $s$ and $t$ belonged to the set of the vertices identified by $w$). 
However, then we can replace $P$ by the path $xsvty$ to find again that $G$ already contained a copy of $H$ before Rule~\ref{r-identify} was applied.
This copy was induced since $s,t$ were not adjacent, as otherwise, 
$u,v,s,t$ would have induced a $K_4$. 
Hence, we obtained a contradiction.\dia

\medskip
\noindent
Let $u\in A_2$. We let $B(u)$ be the set of neighbours of $u$ that have colour~3 in their list. 

\clm{\label{c-uv}
For every $u\in A_2$, it holds that $B(u)\neq \emptyset$ and 
$B(u)\subseteq N_2\cup N_3.$}{\it Proof of Claim~\ref{c-uv}.}
By Rule~\ref{r-always2}, there is a vertex $v\in N(u)$ such that $3\in L(v)$. 
Vertex $v$ cannot be in $N_1$; otherwise the edge $uv$ implies that $v\in A_1$ and thus $v$ would have list $\{1,2\}$.
This means that $v$ must be in $N_2\cup N_3$.
\dia

\medskip\noindent
We will use the following rule (in Claim~\ref{c-uNeighbors} we show that the colour~$q$ is unique).

\begin{enumerate}[\bf Rule 1]
\setcounter{enumi}{12}
\item \label{r-three} {\bf ($\mathbf{A_2}$ list-reduction)}
If a vertex $v\in B(u)$ for some $u\in A_2$ has no neighbour outside $N[u]$, then remove colour $q$ from $L(u)$ for $q\in L(v)\setminus \{3\}$.
\end{enumerate}
\clm{\label{c-uNeighbors}
Rule~\ref{r-three} is safe, takes polynomial time and does not affect any vertex of $N_0$.
Moreover, afterwards $G$ is a connected $(H, K_4)$-free graph, in which
we can define sets $N_1,N_2,N_3,A_1,A_2$ as before.}{\it Proof of Claim~\ref{c-uNeighbors}.}
Let $u$ be a vertex in $A_2$ for which there exists a vertex $v\in B(u)$ with no neighbour outside $N[u]$. It is readily seen that Rule~\ref{r-three} applied on $u$ takes polynomial time, does not affect any vertex of $N_0$, and afterwards we can define sets $N_1,N_2,N_3,A_1,A_2$ as before.

We recall by Claim~\ref{c-uv} that $v\in N_2\cup N_3$. As $N(v)\setminus N[u]=\emptyset$, we find by Rule~\ref{r-identify} that $|L(v)|\neq 3$. Then, by Rule~\ref{r-one}, it holds that $|L(v)|=2$.  Thus vertex $v$ has $L(v)=\{q,3\}$ for some $q\in{\{1,2\}}$. If there exists a colouring $c$ of $G$ with $c(u)=q$ that respects $L$, then $c(v)=3$, and so $c$ colours each vertex in $N(v)\cap N(u)$ with a colour from $\left\{ 1,2 \right\}$. 

We define a colouring $c'$ by setting $c'(u)=3$, $c'(v)=q$ and $c'= c$ for $V(G)\setminus \{u,v\}$.  We claim that $c'$ also respects $L$. As $N(v)\setminus N[u]=\emptyset$, every neighbour $w\neq u$ of $v$ is a neighbour of $u$ as well and thus received a colour $c'(w)=c(w)$ that is not equal to colour $q$ (and colour~$3$). As $v\in N_2\cup N_3$ by Claim~\ref{c-uv}, all vertices in $N(u)\setminus N[v]$ are in $N_1$ by Claim~\ref{c-complete}. As $u\in A_2$, these vertices all belong to $A_1$ and thus their lists are equal to $\{1,2\}$, so do not contain colour 3. Hence, $c'$ respects $L$ indeed.  

The above means that  we can avoid assigning colour~$q$ to $u$. We may therefore remove~$q$ from $L(u)$. This completes the proof of the claim. \dia

\medskip
\noindent
We note that if a colour~$q$ is removed from the list of some vertex~$u\in A_2$ due to Rule~\ref{r-three}, then $u$ is no longer active. 

Assume that Rules~\ref{r-empty}--\ref{r-three} have been applied exhaustively.
By Rule~\ref{r-2sat}, we find that $A_2\neq \emptyset$.
Then we continue as follows.
Let $u\in A_2$ and $v\in B(u)$ (recall that $B(u)$ is nonempty due to Claim~\ref{c-uv}).
Let $A(u,v)\subseteq N_1$ be the set of (active) neighbours of $u$ that are not adjacent to $v$. Note that $A(u,v)\subseteq A_1$ by definition.
Let $A(v,u)\subseteq N_1$ be the set of neighbours of $v$ that are not adjacent to $u$.
We claim that both $A(u,v)$ and $A(v,u)$ are nonempty. 
By Rule~\ref{r-three}, we find that $A(v,u)\neq \emptyset$. By Rule~\ref{r-identify}, vertex $u$ has a neighbour $t\notin N(v)$. As $v\in N_2\cup N_3$ due to Claim~\ref{c-uv}, we find by Claim~\ref{c-complete} that $t$ belongs to $N_1$, thus $t\in A(u,v)$, and consequently, $A(u,v)\neq \emptyset$.
We have the following three disjoint situations:
\begin{enumerate}
\item $A(v,u)$ contains a vertex $w$ with $L(w)=\{1,2\}$ that is not adjacent to some vertex $t\in A(u,v)$;
\item $A(v,u)$ contain at least one vertex~$w$ that is not adjacent to some vertex $t\in A(u,v)$, but for all such vertices~$w$ it holds that $L(w)\neq \{1,2\}$.
\item Every vertex in $A(v,u)$ is adjacent to every vertex of $A(u,v)$.
\end{enumerate}
Now we construct a triple $(Q,P,x)=(Q(u),P(u),x(u))$ such that $Q$ is a set which contains $u$, $P\subseteq Q$ is an induced $P_4$ and $x$ is a vertex of $Q$. 
In Situation 1, we let $Q=\{w,t,u,v\}$. We say that $Q$ is of Type~1. We let $x=u$. As $P$ we can take the path on vertices $t,u,v,w$ in that order. 
In Situation 2, we let $Q=\{w,t,u,v\}$ for some $w\in A(v,u)$ that is not adjacent to some $t\in A(u,v)$. We say that $Q$ is of Type~2. We let $x=v$.
As $P$ we can take the path on vertices $t,u,v,w$ in that order. 

Finally, we consider Situation~3.
Let $w$ be in $A(v,u)$. 
Recall that $u$ is active, $|L(w)|=2$, and in Situation~3 all vertices of $A(u,v)$ are adjacent to all vertices in $A(v,u)$, and thus in particular to~$w$.
Therefore, $u$ has a neighbour $s\notin A(u,v)$ that is not adjacent to $w$, otherwise Rule~\ref{r-hood} would be used, a contradiction with $u$ being active.
If $s$ is in $N_1$, then $s$ is adjacent to $v$ since $s$ is not in $A(u,v)$.
If $s$ is in $N_2\cup N_3$, then $s$ is adjacent to $v$ by Claim~\ref{c-layers}. 
Hence, in both cases we find that $s$ belongs to $N(u)\cap N(v)$.

We let $Q=\{s,t,w,u,v\}$ for some $t\in A(u,v)$. We let $x=v$.
We say that $Q$ is of Type~3.
We claim that the vertices $s,u,t,w$ induce a $P_4$ in that order. By definition, $u$ is not adjacent to $w$.
If $sw\in E(G)$, then $L(u)=L(w)$ due to
Rule~\ref{r-diamond}. As $w$ has a list of size~2, $u$ has also a list of size~2. This is a contradiction, as $u$ is an active vertex. 
If $st\in E(G)$, then $L(v)=L(t)$  due to Rule~\ref{r-diamond}. However, this is also a contradiction, as $L(t)=\{1,2\}$ (since $t\in A_1$) and $3\in L(v)$.
Hence, as $P$ we can take the path on vertices $s,u,t,w$ in that order.

In all three situations, we try to extend $Q$ as follows. If $A(u,v)$ contains more vertices than only vertex~$t$, we pick an arbitrary vertex $t'$ of $N(u)\cap N_1\setminus \{t\}$ and put $t'$ to $Q$.

We first observe that if $c(x)=3$ no other vertex of $Q$ can be coloured with colour~$3$; in particular recall that $t$ and $t'$ (if $t'$ exists) both belong to $A_1$, and as such have list~$\{1,2\}$. Moreover, if $Q$ is of Type~2, then any vertex in $A(v,u)$ with list $\{1,2\}$ is adjacent to $t$, as otherwise $Q$ is of Type~1.

\medskip
\noindent
{\bf Branching VI} ($O(n)$ branches)\\
We choose a vertex $u\in A_2$ such that $|N(u)\cap N_1|$ is minimal and create $(Q,P,x)$.
We branch by considering all possibilities of colouring $Q$ such that $c(x)=3$ and the possibility where we remove colour~$3$ from $L(x)$.
The first case leads to $O(1)$ branches, since $|Q|\le6$. We will prove that we either terminate by Rule~\ref{r-2sat} or branch in Branching VII.
In the second case we deactivate $u$ directly or by applying Rules~\ref{r-three} and~\ref{r-one2}.
This is the only recursive branch and the depth of the recursion is $|A_2|\in O(n)$.
Since the first case in the recursion tree always leads either to termination or to subsequent branching in Branching VII, the branching tree in Branching VI  
can be seen as a path of length $O(n)$, where at each node $O(1)$ branches are created. Hence, we have a total of $O(n)$ branches in Branching VI.

\begin{figure}[tb]
	\begin{center}
		\includestandalone[mode=image]{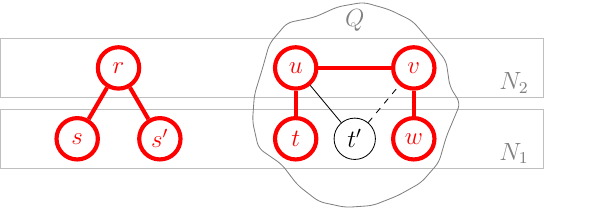}
		\caption{The situation in Branching VI for $Q$ of Type~1. Dashed line denotes an edge that might or might not be there.}\label{f-br6*}
	\end{center}
		\small
\end{figure}

Now consider a branch where $Q$ is coloured.
Although by Rule~\ref{r-one} vertices in~$Q$ will need to be removed from $G$, we make an exception by temporarily keeping $Q$ in the graph after we coloured it until the end of Branching VII. The reason is that this will be helpful for analysing the structure of $(G,L)$.
We run only Rules~\ref{r-2sat},~\ref{r-one2} and~\ref{r-triangle} to prevent changes in the size of neighbourhood of vertices in $A_2$ for the purposes of the next claim (Claim~\ref{c-rLargeDegree}).
Observe that Rules~\ref{r-2sat},~\ref{r-one2} and~\ref{r-triangle} do not decrease the degree of any vertex.
By Rule~\ref{r-2sat}, $A_2\neq \emptyset$. We prove the following claim for vertices in $A_2$.

\clm{\label{c-rLargeDegree}
There is no vertex in $A_2$ with more than one neighbour in $A_1$. Moreover $N(u)\cap A_1=\emptyset$.}{\it Proof of Claim~\ref{c-rLargeDegree}.}
For contradiction, assume that $r$ is a vertex in $A_2$ with two or more neighbours in $A_1$.
By Rule~\ref{r-triangle}, any two distinct neighbours of $r$ in $A_1$ are not adjacent, that is, the neighbours of $r$ in $A_1$ form an independent set.
In particular, for any two distinct neighbours $s$ and $s'$ of $r$ in $A_1$,
the set $\{s,r,s'\}$ induces a~$P_3$. We denote such a path by $P_{s,s'}'$. As every vertex in $A_1$ has list $\{1,2\}$, the only possible edges between $Q$ and $P_{s,s'}'$ are those between $\{s,s'\}$ and vertex~$x$, the only vertex in $Q$ which has colour $3$.

First suppose that $Q$ is of Type~1. Recall that $x=u$.
If $t'$ does not exists, meaning $|N(u)\cap N_1|=1$, the claim follows.
Suppose there exist at least two coloured vertices $t,t'\in Q\cap N(u)\cap N_1$. 
Observe that $N(r)\cap A_1=N(r)\cap N_1$.
We know that $u$ is adjacent to all but one vertex in $N(r)\cap A_1$, as otherwise there are at least two vertices $s$ and $s'$ in $N(r)\cap (A_1\setminus N(u))$ and therefore $V(P'_{s,s'})\cup Q$ induces a $P_3+P_4$, which would be a contradiction.
This situation is captured in Figure~\ref{f-br6*}.
Hence, we find that $|N(u)\cap N_1| \geq |N(r)\cap N_1|-1+2$, which contradicts the choice of $u$.
Thus, if $Q$ is of Type~1, $|N(u)\cap N_1|=1$, so $N(u)\cap A_1=\emptyset$.

Now suppose that $Q$ is of Type~2. Recall that $x=v$. Recall also that if $v$ is adjacent to a vertex in $A_1$, then this vertex must be adjacent to another vertex from $Q$ (either $u$ or $t$) as well, since otherwise $Q$ would be of Type~1.
This is not possible since all vertices in $Q$ are already coloured by colour in $\{1,2\}$.
Therefore we obtain an induced $P_3+P_4$, a contradiction.

Finally, suppose that $Q$ is of Type~3. Recall that $x$ is not in $P$, thus there is no vertex with a list $\{1,2\}$ adjacent to $P$.
Therefore we obtain an induced $P_3+P_4$, a contradiction.

If $Q$ is of Type~2 or~3, vertex $u$ obtained a colour from the set $\{1,2\}$. 
Hence, $N(u)\cap A_1=\emptyset$.
\dia

\medskip 
\noindent
We now run reduction Rules~\ref{r-empty}--\ref{r-three} exhaustively (and in the right order). Recall, however, that we make an exception by not deleting the vertices of $Q$ (specifically, we do not perform the Rule~\ref{r-identify} if it would involve identification or deletion of a vertex in $Q$).

\rmk{\label{c-a2SingleNeigh}
Claim~\ref{c-rLargeDegree} still holds after Rules~\ref{r-empty}--\ref{r-three} were applied.}{\it Proof of the remark.
}
All vertices in $A_2$ had exactly one neighbour before applying Rules~\ref{r-empty}--\ref{r-three}, by Claim~\ref{c-rLargeDegree}. 
It is readily seen that only Rule~\ref{r-identify} can increase the degree of vertices and no rule can increase the size of a list of any vertex.
This  implies that $N(u)\cap A_1=\emptyset$.

For contradiction, assume that there is a vertex $r$ in $A_2$ with more than one neighbour in $A_1$.
Vertex $r$ was created by Rule~\ref{r-identify}, i.e., by identification of at least two vertices $r_1$, $r_2$ which are common neighbours of two adjacent vertices $s, s'$ satisfying the assumptions of Rule~\ref{r-identify}, in particular $|L(s')|=3$.
Observe that $|L(r_1)|=|L(r_2)|=3$, as $|L(r)|=3$ and $L(r)=L(r_1)\cap L(r_2)$ by Rule~\ref{r-identify}.
Therefore, $r_1,r_2,s'\in N_2\cup N_3$. 
Vertices $r_1, r_2$ are non-adjacent, otherwise $s, s', r_1, r_2$ is a $K_4$.
This is a contradiction with Claim~\ref{c-complete}, as $r_1,r_2,s'$ are not a clique.
\dia

\medskip
\noindent
{\bf Branching VII} ($O(n)$ branches)\\
We branch by considering the possibility of removing colour~3 from the list of each vertex in $A_2$, 
and all possibilities of choosing one vertex in $A_2$, to which we give colour~3, and all possibilities of colouring its neighbour in $A_1$ (recall that this neighbour is unique due to Claim~\ref{c-rLargeDegree}). This leads to $O(n)$ branches. We show that all of them are instances with no vertex with list of size~$3$ and thus
Rule~\ref{r-2sat} can be applied on them.

In the first branch, all lists have size at most~$2$ directly by the construction. 
 
Now consider a branch where a vertex $r\in A_2$ and its unique neighbour $r_1$ in $A_1$ were coloured (where $r$ is given colour~3). 
We make an exception to Rule~\ref{r-one} and temporarily keep vertex $r$ and all its neighbours in $G$, even if they need to be removed from $G$ due to our rules. 

Recall that before $r_1$ was coloured, $L(r_1)=\{1,2\}$ and that every vertex in $A_2$ has exactly one neighbour in $A_1$.
Before assigning a colour to $r$, vertex $r$ had exactly two other neighbours $r_2$ and $r_3$ by Rule~\ref{r-always}, which were in $N_2\cap N_3$, and which were adjacent by Claim~\ref{c-complete}.
We claim that $\{r_1,r,r_2\}$ and $\{r_1,r,r_3\}$ induce a $P_3$, as otherwise $\{r,r_1,r_2,r_3\}$ induce a $K_4$ or a diamond: the first case is not possible due to $K_4$-freeness and in the second case we would have applied Rule~\ref{r-identify} on $r$  and $r_2$ (if $rr_2$ is an edge), or on $r$ and $r_3$ (if $rr_3$ is an edge).
As $G$ is $(P_3+P_4)$-free, there must be at least one edge between $P$ and $\{r_1,r,r_2\}$ and between $P$ and
 $\{r_1,r,r_3\}$.
We first show that such an edge is not incident to $r_1$.

If there exists an edge between $r_1$ and a vertex from $P$, then this vertex must be $x$ (as $r_1$ was in $A_1$ and $L(r_1)=\{1,2\}$ before it was coloured).
First, suppose $Q$ is of Type~1.
Recall that $x=u$. 
However, by Claim~\ref{c-rLargeDegree} $N(u)\cap A_1=\emptyset$.
Now suppose $Q$ is of Type~2.
Then $x=v$. If $r_1$ is adjacent to $v$, then $r_1$ is adjacent to another vertex in $Q$, a contradiction. 
Finally, suppose that $Q$ is of Type~3. Then $x$ is not in $P$. Thus $r_1$ is not adjacent to $P$. 
We conclude from the above that
both $r_2$ and $r_3$ have a neighbour in $P$.
We now prove the following claim.

\clm{\label{c-colored} All vertices $r, r_1, r_2, r_3$ are coloured: $r$ received colour $3$, and each of $r_1, r_2, r_3$ received either colour $1$ or $2$.}{\it Proof of Claim~\ref{c-colored}.}
We only have to show the claim for vertices $r_2$ and $r_3$.
Recall that both $r_2$ and $r_3$ have a neighbour in $P$.
We claim that neighbourhoods of $r_2$ and $r_3$ in $Q$ are disjoint.
Otherwise $r$, $r_2$, $r_3$ and a common neighbour
$d$ of $r_2$ and $r_3$ in $P$ form a diamond such that $d\in Q$ is coloured, and therefore $r$ was not active due to Rule~\ref{r-diamond}, a contradiction.
Hence, at least one neighbour of $r_2$ or at least one neighbour of $r_3$ has obtained a colour different from~$3$.
Since $r$ is coloured by~$3$, the lists of $r_2$ and $r_3$ were reduced by Rule~\ref{r-one2} to $\{1\}$ or $\{2\}$ (or the instance is a no-instance). 
\dia

\begin{figure}[tb]
	\begin{center}
		\includestandalone[mode=image,scale=0.8]{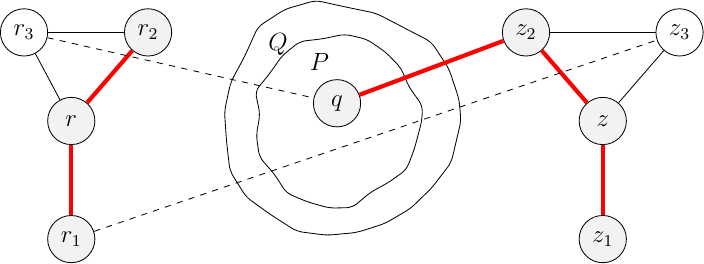}
		\caption{The situation in Branching VII. The dashed lines denote edges that might or might not be there.}\label{f-br7*}
	\end{center}
		\small
\end{figure}

\medskip
\noindent
We are now ready to show that no vertex has a list of size~$3$, and thus applying  Rule~\ref{r-2sat} will solve the instance.
For contradiction assume that there exists a vertex $z$ with $|L(z)|=3$, that is, $z\in A_2$.
Vertices $z_1,z_2,z_3\in N(z)$ exist as $z\in A_2$.
Those vertices are disjoint from $r,r_1,r_2,r_3$ which are by Claim~\ref{c-colored} coloured since $|L(z)|=3$.
The same observations as for neighbours of $r$ hold for neighbours of $z$ by the same arguments as above. Namely, vertex $z_1\in N(z)\cap A_1$ does not have a neighbour in~$P$ and vertices $z_2,z_3$ are in $N_2\cup N_3$ and they induce two $P_3$s: $z_1,z,z_2$ and $z_1,z,z_3$.
Therefore, $z_2,z_3$ have disjoint neighbourhoods in $P$.
Moreover, at least one edge between $r_1$ and $z_2,z_3$ is missing by Rule~\ref{r-diamond} applied on $r_1,z,z_2,z_3$.
We may assume without loss of generality that $r_1z_2\notin E$. Then vertices $z_1,z,z_2,q$, where $q$ is in $N(z_2)\cap V(P)$, induce a new $P_4$.
Again at least one vertex from $r_2,r_3$ is not adjacent to $q$, without loss of generality assume that $r_2q \notin E$, as $r_2$ and $r_3$ have disjoint neighbourhoods in $P$.
As $r_1$ and $r_2$ are coloured by~1 or~2 by Claim~\ref{c-colored}, they have no edge to $z_1$ and to $z$; otherwise $z$ and $z_1$ are not active by Rule~\ref{r-one2}. Recall that $r_1, z_1$ have no neighbour in $P$ and that $r$ had only one neighbour in $A_1$, thus $r$ is not adjacent to $z_1$. 
By~Claim~\ref{c-complete} there are no edges between $r,r_2, r_3$ and $z,z_2, z_3$. Hence $r_1,r,r_2$ together with $z_1,z,z_2,q$ induce a $P_3+P_4$ in $G$, a contradiction (see Figure~\ref{f-br7*} for an example of such a situation).

\medskip
\noindent
The correctness of our algorithm follows from the above description. It remains to analyse its running time. The branching is done in seven stages (Branching I-VII) yielding a total number of $O(n^{49})$ branches. It is readily seen that processing each branch created in Branching I-VII takes polynomial time. 
Hence the total running time of our algorithm is polynomial. 
\end{proof}

\medskip
\noindent
{\bf Remark.}
Except for Phase 4 of our algorithm, all arguments in our proof hold for $(P_3+P_5)$-free graphs. The difficulty in Phase 4 is that in contrary to the previous phases we cannot use the vertices from $N_0$ to find an induced $P_3+P_5$ and therefore obtain the contradiction similarly to the previous phases.

\section{The Proof of Corollary~1}\label{a-a}

By combining our new results from Section~\ref{s-poly} with known results from the literature we can now prove Corollary~\ref{c-summary}.

\medskip
\noindent
{\bf Corollary~\ref{c-summary} (restated).}
{\it Let $H$ be a graph with $|V(H)|\leq 7$. If $H$ is a linear forest, then {\sc List $3$-Colouring} is polynomial-time solvable for $H$-free graphs; otherwise already $3$-{\sc Colouring} is \NP-complete for $H$-free graphs.}

\begin{proof}
If $H$ is not a linear forest, then $H$ contains an induced claw or a cycle, which means that
{\sc 3-Colouring} is \NP-complete due to results in~\cite{EHK98,Ho81,LG83}. Suppose $H$ is a linear forest. We first recall that {\sc List 3-Colouring} is polynomial-time solvable for $P_7$-free graphs~\cite{CMSZ17} and thus for $(rP_1+P_7)$-free graphs for every integer $r\geq 0$~\cite{BGPS12b,GJPS}. Now suppose that $H$ is not an induced subgraph of $rP_1+P_7$ for any $r\geq 0$.   If $H=P_1+3P_2$, then the class of $H$-free graphs is a subclass of $4P_3$-free graphs, for which  {\sc List 3-Colouring} is polynomial-time solvable~\cite{BGPS12b,GJPS}. Otherwise, $H$ has at least two connected components, all of which containing at least one edge. This means that $H\in \{2P_2+P_3, P_2+P_5, P_3+P_4\}$. If $H=2P_2+P_3$, then the class of $H$-free graphs is a subclass of $4P_3$-free graphs, for which we just recalled that {\sc List 3-Colouring} is polynomial-time solvable. The cases where $H=P_2+P_5$ and $H=P_3+P_4$ follow from Theorem~\ref{t-main}.
\end{proof}

\section{Conclusions}\label{s-con}

By solving two new cases we completed the complexity classifications of {\sc $3$-Colouring} and {\sc List $3$-Colouring} on $H$-free 
graphs for graphs~$H$ up to seven vertices. We showed that both problems become polynomial-time solvable if $H$ is a linear 
forest, while they stay \NP-complete in all other cases.
Chudnovsky et al.~improved our results in a recent arXiv paper~\cite{CHSZ18} that appeared after our paper by showing that
{\sc List $3$-Colouring} is polynomial-time solvable on $(rP_3+P_6)$-free graphs for any $r\geq 0$. In the same paper, they also proved that {\sc $5$-Colouring} is 
\NP-complete for $(P_2+P_5)$-free graphs.
Recall that {\sc $k$-Colouring} $(k\geq 3)$ is \NP-complete on $H$-free graphs whenever $H$ is not a linear forest. For the case where $H$ is a linear forest, the \NP-hardness result of~\cite{CHSZ18} for {\sc 5-Colouring} for $(P_2+P_5)$-free graphs, together with the known \NP-hardness results of~\cite{Hu16} for {\sc 4-Colouring} for $P_7$-free graphs and {\sc 5-Colouring} for $P_6$-free graphs, bounds the number of open cases of $k$-{\sc Colouring} from above.

For future research, we remark that it is still not known if there exists a linear forest~$H$ such that {\sc $3$-Colouring} is \NP-complete for $H$-free graphs.  
This is a notorious open problem studied in many papers; for a recent discussion see~\cite{GOPSSS18}. It is also open for {\sc List 3-Colouring}, where an affirmative answer to one of the two problems yields an affirmative answer to the other one~\cite{GPS14b}.
In the line of our proof method, we pose the question if {\sc $3$-Colouring} is polynomial-time solvable on $(P_2+P_{t-2})$-free graphs for some $t\geq 3$ whenever {\sc $3$-Colouring} is polynomial-time solvable for $P_t$-free graphs.  

For $k\geq 4$, we emphasize that all open cases involve linear forests $H$ whose connected components are small. For instance, if $H$ has at most six vertices, then the polynomial-time algorithm for 4-{\sc Precolouring Extension} on $P_6$-free graphs~\cite{CSZ} implies that there are only three graphs~$H$ with $|V(H)|\leq 6$ for which we do not know the complexity of 4-{\sc Colouring} on $H$-free graphs, namely $H \in\{P_1+P_2+P_3, P_2+P_4,2P_3\}$ (see~\cite{GJPS}).

The main difficulty to extend the known complexity results is that hereditary graph classes characterized by a forbidden induced linear forest are still not sufficiently well understood due to their rich structure (proofs of algorithmic results for these graph classes are therefore often long and technical; see also, for example, \cite{BCMSZ,CSZ}). We need a better understanding of these graph classes in order to make further progress. This is not only the case for the two colouring problems in this paper. For example, the {\sc Independent Set} problem is known to be polynomial-time solvable for $P_6$-free graphs~\cite{GKPP17}, but it is not known if there exists a linear forest~$H$ such that it is \NP-complete for $H$-free graphs. A similar situation holds for {\sc Odd Cycle Transversal} and {\sc Feedback Vertex Set} and a whole range of other problems; see~\cite{BDFJP} for a survey.

\medskip
\noindent
{\bf Acknowledgments.} We thank Karel Kr\'al for pointing out a mistake in a preliminary version of our paper.
We also thank two anonymous reviewers for their detailed comments and suggestions.

\bibliographystyle{plainurl}

\end{document}